\documentclass[11pt]{article}

\def\colorful{1}

\usepackage{amsthm,amsfonts,amsmath,amssymb,epsfig,color,float,graphicx,verbatim}
\usepackage{multirow}

\newif\ifhyper\IfFileExists{hyperref.sty}{\hypertrue}{\hyperfalse}
\hyperfalse
\hypertrue
\ifhyper\usepackage{hyperref}\fi

\usepackage{enumitem}
\usepackage[capitalize]{cleveref} 

\def\colorful{0}

\def\nnewcolor{1}
\ifnum\nnewcolor=1

\fi
\ifnum\nnewcolor=0

\fi

\ifnum\colorful=1
\newcommand{\new}[1]{{\color{red} #1}}

\else
\newcommand{\new}[1]{{#1}}

\fi

\usepackage{multirow}
\usepackage{amsmath, color, enumitem}
\usepackage{framed}
\usepackage{nicefrac}

\newtheorem{theorem}{Theorem}[section]

\newtheorem{lemma}[theorem]{Lemma}
\newtheorem{proposition}[theorem]{Proposition}

\newtheorem{fact}[theorem]{Fact}

\theoremstyle{definition}
\newtheorem{definition}[theorem]{Definition}

\newcommand{\R}{\mathbb{R}}

\newcommand{\Z}{\mathbb{Z}}
\newcommand{\E}{\mathbb{E}}

\newcommand{\poly}{\mathrm{poly}}
\newcommand{\polylog}{\mathrm{polylog}}

\newcommand{\dtv}{d_{\mathrm TV}}

\newcommand{\ignore}[1]{}

\newcommand{\eps}{\epsilon}

\newcommand{\eqdef}{\stackrel{{\mathrm {\footnotesize def}}}{=}}

\newcommand{\littlesum}{\mathop{\textstyle \sum}}

\title{Near-Optimal Closeness Testing of Discrete Histogram Distributions}

\author{
Ilias Diakonikolas\thanks{Supported by NSF Award CCF-1652862 (CAREER) and a Sloan Research Fellowship.}\\
University of Southern California\\
{\tt diakonik@usc.edu}\\
\and
Daniel M. Kane\thanks{Supported by NSF Award CCF-1553288 (CAREER) and a Sloan Research Fellowship.}\\
University of California, San Diego\\
{\tt dakane@cs.ucsd.edu}\\
\and
Vladimir Nikishkin\thanks{Supported by a University of Edinburgh PCD Scholarship.}\\
University of Edinburgh\\
{\tt v.nikishkin@sms.ed.ac.uk}
}

\begin{document}

\maketitle

\thispagestyle{empty}

\begin{abstract}
We investigate the problem of testing the equivalence between two discrete histograms.
A {\em $k$-histogram} over $[n]$ is a probability distribution that is piecewise constant over some set of $k$ intervals over $[n]$.
Histograms have been extensively studied in computer science and statistics.
Given a set of samples from two $k$-histogram distributions $p, q$ over $[n]$,
we want to distinguish (with high probability) between the cases that $p = q$ and $\|p-q\|_1 \geq \eps$.
The main contribution of this paper is a new algorithm for this testing problem
and a nearly matching information-theoretic lower bound. 
Specifically, the sample complexity of our algorithm matches our lower bound up to a logarithmic factor, improving
on previous work by polynomial factors in the relevant parameters.
Our algorithmic approach applies in a more general setting and yields improved sample upper bounds 
for testing closeness of other structured distributions as well.
\end{abstract}


\section{Introduction}  \label{sec:intro}
In this work, we study the problem of testing equivalence (closeness) between two discrete {\em structured} distributions.
Let $\mathcal{D}$ be a family of univariate distributions over $[n]$ (or $\Z$).
The problem of {\em closeness testing for $\mathcal{D}$} is the following:
Given sample access to two unknown distribution $p, q \in \mathcal{D}$,
we want to distinguish between the case that $p = q$ versus $\|p-q\|_1 \ge \eps.$
(Here,  $\|p-q\|_1$ denotes the $\ell_1$-distance between the distributions $p, q$.)
The sample complexity of this problem depends on the underlying family
$\mathcal{D}$. 

For example, if $\cal D$ is the class
of {\em all} distributions over $[n],$ then it is known~\cite{CDVV14}
that the optimal sample complexity is $\Theta(\max\{ n^{2/3}/\eps^{4/3}, n^{1/2}/\eps^2\}).$
This sample bound is best possible only if the family $\mathcal{D}$ includes all possible
distributions over $[n]$, and we may be able to obtain significantly better upper bounds 
for most natural settings. For example,
if both $p, q$ are promised to be (approximately) log-concave over $[n]$, there is an algorithm
to test equivalence between them using $O(1/\eps^{9/4})$ samples~\cite{DKN:15:FOCS}.
This sample bound is independent of the support size $n$,
and is dramatically better than the worst-case tight bound~\cite{CDVV14} when $n$ is large.

More generally,~\cite{DKN:15:FOCS} described a framework to obtain
sample-efficient equivalence testers for various families of
structured distributions over both continuous and discrete domains.
While the results of~\cite{DKN:15:FOCS} are sample-optimal
for {\em some} families of distributions (in particular, over continuous domains),
it was not known whether they can be improved for natural families of discrete distributions.
In this paper, we work in the framework
of~\cite{DKN:15:FOCS} and obtain new nearly-matching algorithms and lower bounds. 

Before we state our results in full generality, we describe in detail 
a concrete application of our techniques to the case of {\em histograms} -- 
a well-studied family of structured discrete distributions with a plethora of applications.

\paragraph{Testing Closeness of Histograms.}
A {\em $k$-histogram} over $[n]$ is a probability distribution $p: [n] \to [0, 1]$
that is piecewise constant over some set of $k$ intervals over $[n]$.
The algorithmic difficulty in testing properties of such distributions lies in the fact
that the location and ``size'' of these intervals is a priori unknown.
Histograms have been extensively studied in statistics and computer science.
In the database community, histograms~\cite{JPK+98,CMN98,TGIK02,GGI+02, GKS06, ILR12, ADHLS15} 
constitute the most common tool for the succinct approximation of data.
In statistics, many methods have been proposed to estimate histogram distributions
~\cite{Scott79, FreedmanD1981, Scott:92, LN96, Devroye2004, WillettN07, Klem09} 
in a variety of settings. 

In recent years, histogram distributions have attracted renewed interested from the theoretical computer science community
in the context of learning~\cite{DDS12soda, CDSS13, CDSS14, CDSS14b, DHS15, AcharyaDLS16, ADLS17, DiakonikolasKS16a} and 
testing~\cite{ILR12, DDSVV13, DKN:15, Canonne16, CDGR16}. 
Here we study the following testing problem: Given sample access to two distributions $p, q$ over $[n]$
that are promised to be (approximately) $k$-histograms, distinguish between the cases that $p=q$ versus $\|p-q\|_1 \geq \eps$. 
As the main application of our techniques, we give a new testing algorithm and 
a nearly-matching information-theoretic lower bound for this problem.

We now provide a summary of previous work on this problem followed by
a description of our new upper and lower bounds.
We want to $\eps$-test closeness in $\ell_1$-distance between two $k$-histograms over $[n]$, where $k \leq n$.
Our goal is to understand the optimal sample complexity of this problem as a function of $k, n, 1/\eps$.
Previous work is summarized as follows:
\begin{itemize}
\item In~\cite{DKN:15:FOCS}, the authors gave a closeness tester 
with sample complexity $O(\max\{ k^{4/5}/\eps^{6/5}, k^{1/2}/\eps^2 \})$.

\item The best known sample lower bound is $\Omega(\max\{ k^{2/3}/\eps^{4/3}, k^{1/2}/\eps^2\})$.
This straightforwardly follows from~\cite{CDVV14}, since $k$-histograms
can simulate {\em any} support $k$ distribution.
\end{itemize}
Notably, none of the two bounds depends on the domain size $n$. 
Observe that the upper bound of $O(\max\{ k^{4/5}/\eps^{6/5}, k^{1/2}/\eps^2 \})$
cannot be tight for the entire range of parameters. For example, for $n = O(k)$, 
the algorithm of~\cite{ CDVV14} for testing closeness between arbitrary support $n$ 
distributions has sample size $O(\max\{ k^{2/3}/\eps^{4/3}, k^{1/2}/\eps^2\})$, matching the 
above sample complexity lower bound, up to a constant factor.

This simple example might suggest that the $\Omega(\max\{ k^{2/3}/\eps^{4/3}, k^{1/2}/\eps^2\})$ 
lower bound is tight in general. We prove that this is not the case. 
The main conceptual message of our new upper bound and nearly-matching lower bound 
is the following: 
\begin{quote}
{\em The sample complexity of $\eps$-testing closeness between two $k$-histograms over $[n]$ depends in a subtle way 
on the relation between the relevant parameters $k, n$ and $1/\eps$.}
\end{quote}
We find this fact rather surprising
because such a phenomenon does {\em not} occur for the sample complexities of 
closely related problems. Specifically, testing the identity of a $k$-histogram over $[n]$ to a {\em fixed} distribution 
has sample complexity $\Theta(k^{1/2}/\eps^2)$~\cite{DKN:15}; and learning a $k$-histogram over $[n]$ 
has sample complexity $\Theta(k/\eps^2)$~\cite{CDSS14}. Note that both these sample bounds are independent of $n$ and are known to 
be tight for the entire range of parameters $k, n, 1/\eps$.

Our main positive result is a new closeness testing algorithm for $k$-histograms over $[n]$
with sample complexity 
$O\big(k^{2/3}\cdot \log^{4/3} (2+n/k) \log(k) / \eps^{4/3}\big).$
Combined with the known upper bound of \cite{DKN:15:FOCS}, we obtain the sample upper bound of
$$
O\big(\max\big(\min(k^{4/5}/\eps^{6/5},k^{2/3}\log^{4/3}(2+n/k)\log(k)/\epsilon^{4/3}),k^{1/2}\log^2(k)\log\log(k)/\eps^2 \big) \big).
$$
As our main negative result, we prove a lower bound of 
$\Omega(\min( k^{2/3}\log^{1/3}(2+n/k)/\eps^{4/3}, k^{4/5}/\eps^{6/5}))$.
The first term in this expression shows that the ``$\log(2+n/k)$'' factor that appears in the sample 
complexity of our upper bound is in fact necessary, up to a constant power.
In summary, these bounds provide a nearly-tight characterization of the 
sample complexity of our histogram testing problem for the entire range of parameters.

A few observations are in order to interpret the above bounds:
\begin{itemize}
\item When $n$ goes to infinity, the  $O(k^{4/5}/\eps^{6/5})$ upper bound of \cite{DKN:15:FOCS} is tight for $k$-histograms.

\item When $n = \poly(k)$ and $\eps$ is not too small (so that the $k^{1/2}/\eps^2$ term does not kick in),
then the right answer for the sample complexity of our problem is $(k^{2/3}/\eps^{4/3}) \polylog(k)$. 

\item The terms ``$k^{4/5}/\eps^{6/5}$'' and ``$k^{2/3}\log^{4/3}(2+n/k)\log(k)/\epsilon^{4/3}$'' appearing
in the sample complexity become equal when $n$ is exponential in $k$. Therefore, our new 
algorithm has better sample complexity than that of \cite{DKN:15:FOCS}  for all $n \leq 2^{O(k)}.$ 
\end{itemize}

In the following subsection, we state our results in a general setting 
and explain how the aforementioned applications are obtained from them.

\subsection{Our Results and Comparison to Prior Work}  \label{ssec:results}

For a given family $\mathcal{D}$ of discrete distributions over $[n]$, we are interested 
in designing a closeness tester for distributions in $\mathcal{D}$.
We work in the general framework introduced by~\cite{DKN:15, DKN:15:FOCS}.
Instead of designing a different tester for any given family $\mathcal{D}$,
the approach of~\cite{DKN:15, DKN:15:FOCS} proceeds by designing
a generic equivalence tester under a {\em different metric} than the $\ell_1$-distance.
This metric, termed $\mathcal{A}_k$-distance~\cite{DL:01}, where $k \ge 2$ is a positive integer,
interpolates between Kolmogorov distance (when $k=2$)
and the $\ell_1$-distance (when $k=n$). It turns out that, for a range of structured distribution families
$\mathcal{D}$, the $\mathcal{A}_k$-distance can be used as a proxy for the $\ell_1$-distance
for a value of $k \ll n$~\cite{CDSS14}. 
For example, if $\mathcal{D}$ is the family of $k$-histograms over $[n]$, the $\mathcal{A}_{2k}$
distance between them is tantamount to their $\ell_1$ distance. 
We can thus obtain an $\ell_1$ closeness tester for ${\cal D}$ by plugging in the right value of $k$ in a general $\mathcal{A}_k$ closeness tester.

\medskip

To formally state our results, we will need some terminology.

\medskip

\noindent {\bf Notation.}
We will use $p, q$ to denote the probability mass functions
of our distributions. If $p$ is discrete over support $[n]: = \{1, \ldots, n\}$, we denote
by $p_i$ the probability of element $i$ in the distribution.
For two discrete distributions $p, q$, their $\ell_1$ and $\ell_2$ distances are
$\|p -q \|_1 = \sum_{i=1}^n |p_i - q_i|$ and $\|p-q\|_2 = \sqrt{\sum_{i=1}^n (p_i - q_i)^2}$.
Fix a partition of the domain $I$ into disjoint intervals
$\mathcal{I} :=  (I_i)_{i=1}^{\ell}.$ For such a partition $\mathcal{I}$,
the {\em reduced distribution} $p_r^{\mathcal{I}}$ corresponding to $p$ and $\mathcal{I}$ is the discrete distribution over $[\ell]$
that assigns the $i$-th ``point'' the mass that $p$ assigns to the
interval $I_i$; i.e., for $i \in [\ell]$, $p_r^{\mathcal{I}} (i) = p(I_i)$.
Let $\mathfrak{J}_k$ be the collection
of all partitions of the domain $I$ into $k$ intervals. For $p, q: I \to \R_+$ and $k \in \Z_+$,
we define the $\mathcal{A}_k$-distance between $p$ and $q$ by
$\|p-q\|_{\mathcal{A}_k} \eqdef \max_{\mathcal{I} = (I_i)_{i=1}^{k} \in \mathfrak{J}_k} \littlesum_{i=1}^k |p(I_i) - q(I_i)|
= \max_{\mathcal{I}  \in \mathfrak{J}_k} \| p_r^{\mathcal{I} } - q_r^{\mathcal{I} } \|_1.$

\medskip

In this context,~\cite{DKN:15:FOCS} gave a closeness testing algorithm under the $\mathcal{A}_k$-distance
using $O(\max\{ k^{4/5}/\eps^{6/5}, k^{1/2}/\eps^2 \})$ samples.
It was also shown that this sample bound is
information--theoretically optimal (up to constant factors)
for some adversarially constructed continuous distributions, or discrete distributions of support size $n$ sufficiently large
as a function of $k$. These results raised two natural questions: 
(1) What is the {\em optimal} sample complexity of the $\mathcal{A}_k$-closeness testing 
problem as a function of $n, k, 1/\eps$? 
(2) Can we obtain tight sample lower bounds for {\em natural} families of structured distributions?

\medskip

We resolve both these open questions. Our main algorithmic result is the following:

\begin{theorem}\label{thm:main-alg}
Given sample access to distributions $p$ and $q$ on $[n]$ and $\epsilon>0$
there exists an algorithm that takes
$$
O\left(\max\left(\min\left(k^{4/5}/\epsilon^{6/5},k^{2/3}\log^{4/3}(2+n/k)\log(2+k)/\epsilon^{4/3} \right),k^{1/2}\log^2(k)\log\log(k)/\epsilon^2 \right) \right)
$$
samples from each of $p$ and $q$ and distinguishes with $2/3$ probability between the cases that $p=q$
and $\|p-q\|_{\mathcal{A}_k}\geq \epsilon$.
\end{theorem}

As explained in~\cite{DKN:15, DKN:15:FOCS}, using Theorem~\ref{thm:main-alg} 
one can obtain testing algorithms for the $\ell_1$ closeness testing of various
distribution families ${\mathcal{D}}$, by using the $\mathcal{A}_k$ distance 
as a ``proxy'' for the $\ell_1$ distance:

\begin{fact} \label{fact:simple}
For a univariate distribution family ${\mathcal{D}}$ and $\eps>0$, let $k= {k({\mathcal{D}}, \eps)}$ 
be the smallest integer such that for any $f_1, f_2 \in {\mathcal{D}}$ it holds that
$\|f_1-f_2\|_1 \le \|f_1-f_2\|_{{\mathcal A}_k} + \eps/2$. Then there exists an $\ell_1$ closeness testing algorithm 
for $\mathcal{D}$ with the sample complexity of Theorem~\ref{thm:main-alg}.
\end{fact}

\paragraph{Applications.}
Our upper bound for $\ell_1$-testing of $k$-histogram distributions follows from 
the above by noting that for any $k$-histograms $p, q$ we have $\|p-q\|_1 = \|p-q\|_{\mathcal{A}_{2k}}$.
Also note that our upper bound is {\em robust}: it applies even 
if $p, q$ are $O(\eps)$-close in $\ell_1$-norm to being $k$-histograms.

Finally, we remark that our general $\mathcal{A}_k$ closeness tester yields improved
upper bounds for various other families of structured distributions. Consider for example the
case that $\mathcal{D}$ consists of all $k$-mixtures of some simple family (e.g., discrete Gaussians or log-concave),
where the parameter $k$ is large.
The algorithm of~\cite{DKN:15:FOCS} leads to a tester whose sample complexity scales with 
$O(k^{4/5})$, while Theorem~\ref{thm:main-alg} implies a $\tilde{O}(k^{2/3})$ bound.

\medskip

On the lower bound side, we show:

\begin{theorem} \label{thm:main-lb}
Let $p$ and $q$ be distributions on $[n]$ and let $\epsilon>0$ be less than a sufficiently small constant.
Any tester that distinguishes between $p=q$ and $\|p-q\|_{\mathcal{A}_k}\geq \eps$ for some $k\leq n$
must use $\Omega(m)$ samples for $m=\min(k^{2/3}\log^{4/3}(2+n/k)/\eps^{4/3},k^{4/5}/\eps^{6/5})$.

Furthermore, for $m=\min(k^{2/3}\log^{1/3}(2+n/k)/\eps^{4/3},k^{4/5}/\eps^{6/5})$, 
any tester that distinguishes between $p=q$ and $\|p-q\|_{\mathcal{A}_k}\geq \eps$ 
must use $\Omega(m)$ samples even if $p$ and $q$ are both guaranteed 
to be piecewise constant distributions on $O(k+m)$ pieces.
\end{theorem}

Note that a lower bound of $\Omega(\sqrt{k}/\epsilon^2)$ straightforwardly applies 
even for $p$ and $q$ being $k$-histograms. This dominates the above bounds for $\epsilon < k^{-3/8}$.

We also note that our general lower bound with respect to the $\mathcal{A}_k$ distance is somewhat stronger,
matching the term ``$\log^{4/3}(2+n/k)$'' in our upper bound.


\subsection{Related Work} \label{ssec:related}
During the past two decades,  {\em distribution property testing}~\cite{BFR+:00}
-- whose roots lie in statistical hypothesis testing~\cite{NeymanP, lehmann2005testing} --
has received considerable attention by the computer science community,
see~\cite{Rub12, Canonne15} for two recent surveys.
The majority of the early work in this field has focused on characterizing the sample size needed to test properties
of arbitrary distributions of a given support size. After two decades of study, this ``worst-case''
regime is well-understood: for many properties of interest there exist
sample-optimal testers (matched by information-theoretic lower bounds)
~\cite{Paninski:08, CDVV14, VV14, DKN:15, DK16, DiakonikolasGPP16}.

In many settings of interest, we know a priori that the underlying distributions have some ``nice structure'' (exactly or approximately).
The problem of {\em learning} a probability distribution under such structural assumptions
is a classical topic in statistics, see \cite{BBBB:72} for a classical book, and~\cite{GJ:14} 
for a recent book on the topic, 
that has recently attracted the interest of computer 
scientists~\cite{DDS12soda, DDS12stoc, CDSS13, DDOST13focs, CDSS14, CDSS14b, ADHLS15,
DKS15a, DKS15b, DKS16, DDKT15,  ADLS17, DiakonikolasKS16a, DKS16lcd}.

On the other hand, the theory of {\em distribution testing} under structural assumptions is less fully developed.
More than a decade ago, Batu, Kumar, and Rubinfeld~\cite{BKR:04} considered a specific instantiation of this question --
testing the equivalence between two unknown discrete monotone distributions -- and obtained
a tester whose sample complexity is poly-logarithmic in the domain size. A recent sequence of works~\cite{DDSVV13, DKN:15, DKN:15:FOCS}
developed a framework to leverage such structural assumptions and obtained more efficient testers
for a number of natural settings. However, for several natural properties of interest
there is still a substantial gap between known sample upper and lower bounds. 

\subsection{Overview of Techniques} \label{sec:techniques}

To prove our upper bound, we use a technique of iteratively reducing the number of bins (domain elements).
In particular, we show that if we merge bins together in consecutive pairs,
this does not significantly affect the $\mathcal{A}_k$ distance between the distributions,
unless a large fraction of the discrepancy between our distributions is supported on $O(k)$
bins near the boundaries in the optimal partition. In order to take advantage of this,
we provide a novel identity tester that requires few samples to distinguish between
the cases where $p=q$ and the case where $p$ and $q$ have a large $\ell_1$ distance
supported on only $k$ of the bins. We are able to take advantage of the small support
essentially because having a discrepancy supported on few bins implies
that the $\ell_2$ distance between the distributions must be reasonably large.

Our new lower bounds are somewhat more involved.
We prove them by exhibiting explicit families of pairs of distributions,
where in one case $p=q$ and in the other $p$ and $q$ have large $\mathcal{A}_k$ distance,
but so that it is information-theoretically impossible to distinguish
between these two families with a small number of samples.
In both cases, $p$ and $q$ are explicit piecewise constant distributions
with a small number of pieces. In both cases, our domain is partitioned
into a small number of bins and the restrictions of the distributions
to different bins are independent, making our analysis easier.
In some bins we will have $p=q$ each with mass about $1/m$
(where $m$ is the number of samples). These bins will serve
the purpose of adding ``noise'' making harder to read the ``signal''
from the other bins. In the remaining bins, we will have either that
$p=q$ being supported on some interval, or $p$ and $q$ will be
supported on consecutive, non-overlapping intervals.
If three samples are obtained from any one of these intervals,
the order of the samples and the distributions that they come
from will provide us with information about which family we came from.
Unfortunately, since triple collisions are relatively uncommon,
this will not be useful unless $m\gg \max(k^{4/5}/\epsilon^{6/5},k^{1/2}/\epsilon^{2})$.
Bins from which we have one or zero samples will tell us nothing,
but bins from which we have exactly two samples may provide information.

For these bins, it can be seen that we learn nothing from the ordering of the samples,
but we may learn something from their spacing.
In particular, in the case where $p$ and $q$ are supported on disjoint intervals,
we would suspect that two samples very close to each other are far more likely
to be taken from the same distribution rather than from opposite distributions.
On the other hand, in order to properly interpret this information,
we will need to know something about the scale of the distributions
involved in order to know when two points should be considered to be ``close''.
To overcome this difficulty, we will stretch each of our distributions
by a random exponential amount. This will effectively conceal any information
about the scales involved so long as the total support size of our distributions
is exponentially large.

\section{A Near-Optimal Closeness Tester over Discrete Domains} \label{sec:algo}

\subsection{Warmup: A Simpler Algorithm }

We start by giving a simpler algorithm establishing 
a basic version of Theorem \ref{thm:main-alg} with slightly worse parameters:

\begin{proposition}\label{algorithmProp}
Given sample access to distributions $p$ and $q$ on $[n]$ and $\epsilon>0$
there exists an algorithm that takes
$$
O\left(k^{2/3}\log^{4/3}(3+n/k)\log\log(3+n/k)/\epsilon^{4/3}+\sqrt{k}\log^{2}(3+n/k)\log\log(3+n/k)/\eps^2 \right)
$$
samples from each of $p$ and $q$ and distinguishes with $2/3$ probability between the cases that $p=q$
and $\|p-q\|_{\mathcal{A}_k}\geq \epsilon$.
\end{proposition}

The basic idea of our algorithm is the following:
From the distributions $p$ and $q$ construct new distributions $p'$ and $q'$
by merging pairs of consecutive buckets.
Note that $p'$ and $q'$ each have much smaller domains (of size about $n/2$).
Furthermore, note that the $\mathcal{A}_k$ distance between $p$ and $q$ is
$\sum_{I\in\mathcal{I}}|p(I)-q(I)|$ for some partition $\mathcal{I}$ into $k$ intervals.
By using essentially the same partition, we can show that $\|p'-q'\|_{\mathcal{A}_k}$
should be almost as large as $\|p-q\|_{\mathcal{A}_k}$. This will in fact hold
unless much of the error between $p$ and $q$ is supported at points
near the endpoints of intervals in $\mathcal{I}$.
If this is the case, it turns out there is an easy algorithm to detect this discrepancy.
We require the following definitions:

\begin{definition}
For a discrete distribution $p$ on $[n]$, the merged distribution obtained from $p$ is
the distribution $p'$ on $\lceil n/2\rceil$, so that $p'(i) \eqdef p(2i)+p(2i +1)$.
For a partition $\mathcal{I}$ of $[n]$ , define the \textit{divided partition}
$\mathcal{I}'$ of domain $\lceil n/2\rceil$, so that $I'_i \in \mathcal{I'}$
 has the points obtained by point-wise gluing together odd points and even points.

Note that one can simulate a sample from $p'$ given a sample from $p$ by letting $p'=\lceil p/2 \rceil$.
\end{definition}

\begin{definition}
Let $p$ and $q$ be distributions on $[n]$.
For integers $k\geq 1$, let $\|p-q\|_{1,k}$ be the sum of the largest $k$
values of $|p(i)-q(i)|$ over $i \in [n]$.
\end{definition}

We begin by showing that either $\|p'-q'\|_{\mathcal{A}_k}$
is close to $\|p-q\|_{\mathcal{A}_k}$ or $\|p-q\|_{1,k}$ is large.

\begin{lemma}\label{smallErrorSupportLem}
For any two distributions $p$ and $q$ on $[n]$, let $p'$ and $q'$ be the merged distributions. Then,
$$
\|p-q\|_{\mathcal{A}_k} \leq \|p'-q'\|_{\mathcal{A}_k} + 2\|p-q\|_{1,k} \;.
$$
\end{lemma}
\begin{proof}
Let $\mathcal{I}$ be the partition of $[n]$ into $k$ intervals so that
$\|p-q\|_{\mathcal{A}_k} = \sum_{I\in \mathcal{I}}|p(I)-q(I)|.$
Let $\mathcal{I'}$ be obtained from $\mathcal{I}$ by rounding
each upper endpoint of each interval except for the last down
to the nearest even integer, and rounding the lower endpoint
of each interval up to the nearest odd integer. Note that
$$
\sum_{I\in\mathcal{I'}} |p(I)-q(I)| = \sum_{I\in\mathcal{I'}} |p'(I/2)-q'(I/2)| \leq \|p'-q'\|_{\mathcal{A}_k} \;.
$$
The partition $\mathcal{I'}$ is obtained from $\mathcal{I}$
by taking at most $k$ points and moving them from one interval to another.
Therefore, the difference
$$
\left|\sum_{I\in\mathcal{I}} |p(I)-q(I)|-\sum_{I\in\mathcal{I'}} |p(I)-q(I)| \right| \;,
$$
is at most twice the sum of $|p(i)-q(i)|$
over these $k$ points, and therefore at most $2\|p-q\|_{1,k}$.
Combing this with the above gives our result.
\end{proof}

Next, we need to show that if two distributions have $\|p-q\|_{1,k}$ large that this can be detected easily.

\begin{lemma}\label{BkAlgLem}
Let $p$ and $q$ be distributions on $[n]$.
Let $k>0$ be a positive integer, and $\epsilon>0$.
There exists an algorithm which takes $O(k^{2/3}/\epsilon^{4/3}+\sqrt{k}/\eps^2)$
samples from each of $p$ and $q$ and, with probability at least $2/3$,
distinguishes between the cases that $p=q$ and $\|p-q\|_{1,k}>\epsilon$.
\end{lemma}
Note that if we needed to distinguish between $p=q$ and $\|p-q\|_1 > \epsilon$, 
this would require $\Omega(n^{2/3}/\epsilon^{4/3}+\sqrt{n}/\eps^2)$ samples. 
However, the optimal testers for this problem are morally $\ell_2$-testers. 
That is, roughly, they actually distinguish between $p=q$ and $\|p-q\|_2 > \epsilon/\sqrt{n}$. 
From this viewpoint, it is clear why it would be easier to test for discrepancies in $\| - \|_{1,k}$-distance, 
since if $\|p-q\|_{1,k} > \epsilon$, then $\|p-q\|_2 > \epsilon/\sqrt{k}$, making it easier for our $\ell_2$-type tester to detect the difference.

Our general approach will be by way of the techniques developed in \cite{DK16}. We begin by giving the definition of a split distribution coming from that paper:

\new{
\begin{definition}
Given a distribution $p$ on $[n]$ and a multiset $S$ of elements of $[n]$, define the \emph{split distribution} $p_S$ on $[n+|S|]$ as follows:
For $1\leq i\leq n$, let $a_i$ denote $1$ plus the number of elements of $S$ that are equal to $i$.
Thus, $\sum_{i=1}^n a_i = n+|S|.$ We can therefore associate the elements of $[n+|S|]$ to elements of the set
$B=\{(i,j):i\in [n], 1\leq j \leq a_i\}$.
We now define a distribution $p_S$ with support $B$, by letting a random sample from $p_S$ be given by $(i,j)$,
where $i$ is drawn randomly from $p$ and $j$ is drawn randomly from $[a_i]$.
\end{definition}
We now recall two basic facts about split distributions:
\begin{fact}[\cite{DK16}]\label{splitDistributionFactsLem}
Let $p$ and $q$ be probability distributions on $[n]$, and $S$ a given multiset of $[n]$. Then:
(i) We can simulate a sample from $p_S$ or $q_S$ by taking a single sample from $p$ or $q$, respectively.
(ii) It holds $\|p_S-q_S\|_1 = \|p-q\|_1$.
\end{fact}

\begin{lemma}[\cite{DK16}]\label{splitL2Lem}
Let $p$ be a distribution on $[n]$. Then:
(i) For any multisets $S\subseteq S'$ of $[n]$, $\|p_{S'}\|_2 \leq \|p_S\|_2$, and
(ii) If $S$ is obtained by taking $m$ samples from $p$, then $\E[\|p_S\|_2^2] \leq 1/m$.
\end{lemma}
}


We also recall an optimal $\ell_2$ closeness tester under the promise that one of the 
distributions has smal $\ell_2$ norm:

\begin{lemma}[\cite{CDVV14}] \label{L2TestLem}
Let $p$ and $q$ be two unknown distributions on $[n]$.
There exists an algorithm that on input $n$,  $b \geq \min \{\|p\|_2, \|q\|_2 \}$
and $0< \eps < \sqrt{2}b$, 
draws $O(b/\eps^2)$ samples
from each of $p$ and $q$ and, with probability at least $2/3$,
distinguishes between the cases that $p=q$ and $\|p-q\|_2 > \eps.$
\end{lemma}

\begin{proof}[Proof of Lemma~\ref{BkAlgLem}:]
We begin by presenting the algorithm:

\smallskip

\fbox{\parbox{6in}{
{\bf Algorithm} \texttt{Small-Support-Discrepancy-Tester}\\
Input: sample access to pdf's $p, q: [n] \to [0, 1]$, $k \in \Z_+$, and $\eps > 0$.\\
Output: ``YES'' if $q = p$; ``NO'' if $\|q-p\|_{1,k} \ge \eps.$


\begin{enumerate}

\item Let $m=\min(k^{2/3}/\epsilon^{4/3},k)$.

\item Let $S$ be the multiset obtained by taking $m$ independent samples from $p$.

\item Use the $\ell_2$ tester of Lemma~\ref{L2TestLem} to distinguish between the cases that 
$p_S=q_S$ and $\|p_S-q_S\|_2^2 \geq k^{-1}\epsilon^2/2$ and return the result.

\end{enumerate}
}}

\vspace{0.3cm}

The analysis is simple.
By Lemma~\ref{splitL2Lem}, with $90\%$ probability $\|p_S\|_2 = O(m^{-1/2})$,
and therefore the number of samples needed (using the $\ell_2$ tester from Lemma \ref{L2TestLem})
is $O(m+km^{-1/2}/\epsilon^{2}) = O(k^{2/3}/\epsilon^{4/3}+\sqrt{k}/\eps^2).$
If $p=q$, then $p_S=q_S$ and the algorithm will return ``YES''
with appropriate probability. If $\|q-p\|_{1,k} \ge \eps$, then $\|p_S-q_S\|_{1,k+m}\geq \epsilon$.
Since $k+m$ elements contribute to total $\ell_1$ error at least $\epsilon$,
by Cauchy-Schwarz, we have that $\|p_S-q_S\|_2^2 \geq \epsilon^2/(k+m) \geq k^{-1}\epsilon^2/2.$
Therefore, in this case, the algorithm returns ``NO'' with appropriate probability.
\end{proof}

\begin{proof} [Proof of Proposition \ref{algorithmProp}:]
The basic idea of our algorithm is the following: 
By Lemma \ref{BkAlgLem}, if $\|p-q\|_{\mathcal{A}_k}$ is large, then so is either $\|p-q\|_{1,k}$ or $\|p'-q'\|_{\mathcal{A}_k}$. 
Our algorithm then tests whether $\|p-q\|_{1,k}$ is large, and recursively tests whether $\|p'-q'\|_{\mathcal{A}_k}$ is large. 
Since $p',q'$ have half the support size, we will only need to do this for $\log(n/k)$ rounds, 
losing only a poly-logarithmic factor in the sample complexity.
We present the algorithm here:

\vspace{0.3cm}

\fbox{\parbox{6in}{
{\bf Algorithm} \texttt{Small-Domain-$\mathcal{A}_k$-tester}\\
Input: sample access to pdf's $p, q: [n] \to [0, 1]$, $k \in \Z_+$, and $\eps > 0$.\\
Output: ``YES'' if $q = p$; ``NO'' if $\|q-p\|_{\mathcal{A}_k} \ge \eps.$


\begin{enumerate}

\item For $i:=0$ to $t \eqdef \lceil \log_2(n/k)\rceil$, let $p^{(i)},q^{(i)}$
be distributions on $[\lceil 2^{-i}n \rceil]$ defined by
$p^{(i)}=\lceil 2^{-i}p \rceil$ and $q^{(i)}=\lceil 2^{-i}q \rceil$.

\item Take $Ck^{2/3}\log^{4/3}(3+n/k)\log\log(3+n/k)/\epsilon^{4/3}$ samples, for $C$ sufficiently large,
and use these samples to distinguish between the cases
$p^{(i)}=q^{(i)}$ and $\|p^{(i)}-q^{(i)}\|_{1,k} > \epsilon/(4 \log_2(3+n/k))$
with probability of error at most $1/(10 \log_2(3+n/k))$ for each $i$ from $0$ to $t$, using the same samples for each test.

\item If any test yields that $p^{(i)} \neq q^{(i)}$, return ``NO''. Otherwise, return ``YES''.

\end{enumerate}
}}

\vspace{0.3cm}

We now show correctness. In terms of sample complexity, we note that by taking a majority over $O(\log\log(3+n/k))$ independent
runs of the tester from Lemma \ref{BkAlgLem} we can run this algorithm with the stated sample complexity.
Taking a union bound, we can also assume that all tests performed in Step 2 returned the correct answer.
If $p=q$ then $p^{(i)}=q^{(i)}$ for all $i$ and thus, our algorithm returns ``YES''.
Otherwise, we have that $\|p-q\|_{\mathcal{A}_k} \geq \epsilon$.
By repeated application of Lemma \ref{smallErrorSupportLem}, we have that
$$
\|p-q\|_{\mathcal{A}_k} \leq \sum_{i=0}^{t-1} 2\|p^{(i)}-q^{(i)}\|_{1,k} + \|p^{(t)}-q^{(t)}\|_{\mathcal{A}_k} \leq 2 \sum_{i=0}^{t} \|p^{(i)}-q^{(i)}\|_{1,k} \;,
$$
where the last step was because $p^{(t)}$ and $q^{(t)}$ have a support of size at most $k$
and so $\|p^{(t)}-q^{(t)}\|_{\mathcal{A}_k} = \|p^{(t)}-q^{(t)}\|_1 = \|p^{(t)}-q^{(t)}\|_{1,k}$.
Therefore, if this is at least $\epsilon$, it must be the case that
$\|p^{(i)}-q^{(i)}\|_{1,k}> \epsilon/(4 \log_2(3+n/k))$ for some $0\leq i\leq t$, and thus our algorithm returns ``NO''.
This completes our proof.
\end{proof}

\subsection{Full Algorithm}

The improvement to Proposition \ref{algorithmProp} is somewhat technical. 
The key idea involves looking into the analysis of Lemma \ref{BkAlgLem}. 
Generally speaking, choosing a larger value of $m$ (up to the total sample complexity), 
will decrease the $\ell_2$ norm of $p$, and thus the final complexity. 
Unfortunately, taking $m>k$ might lead to problems as it will subdivide the $k$ original bins 
on which the error is supported into $\omega(k)$ bins. This in turn could worsen 
the lower bounds on $\|p-q\|_2$. However, this will only be the case if the total mass 
of these bins carrying the difference is large. Thus, we can obtain 
an improvement to Lemma \ref{BkAlgLem} when the mass of bins on which the error is supported is small.
This motivates the following definition:

\begin{definition}
For probability distributions $p, q$, an integer $k$ and real number $\alpha>0$, 
$d_{k,\alpha}(p,q)$ is the maximum over sets $T$ of size at most $k$ 
so that $p(i) \leq \alpha$ for all $i\in T$ of $\sum_{i\in T} |p(i)-q(i)|$.
\end{definition}

In other words, $d_{k,\alpha}(p,q)$ is the biggest $\ell_1$ difference between $p$ and $q$ 
coming from at most $k$ bins of mass at most $\alpha$. 
We have the following lemma:

\begin{lemma}\label{BkAlgLem2}
Let $p$ and $q$ be distributions on $[n]$.
Let $k>0$ be a positive integer, and $\epsilon,\alpha>0$.
There exists an algorithm which takes $O(k^{2/3}/\epsilon^{4/3}(1+m\alpha))$
samples from each of $p$ and $q$ and, with probability at least $2/3$,
distinguishes between the cases that $p=q$ and $d_{k,\alpha}(p,q)>\epsilon$.
\end{lemma}
\begin{proof}
The algorithm and its analysis are nearly identical to that of Lemma \ref{BkAlgLem}.
We include them here for completeness:

\smallskip

\fbox{\parbox{6in}{
{\bf Algorithm} \texttt{Small-Support-Discrepancy-Tester}\\
Input: sample access to pdf's $p, q: [n] \to [0, 1]$, $k \in \Z_+$, and $\eps > 0$ with $\|p\|_2\leq \alpha$.\\
Output: ``YES'' if $q = p$; ``NO'' if $\|q-p\|_{1,k} \ge \eps.$


\begin{enumerate}

\item Let $m=k^{2/3}/\epsilon^{4/3}$.

\item Let $S$ be the multiset obtained by taking $m$ independent samples from $p$.

\item Use the $\ell_2$ tester of Lemma~\ref{L2TestLem} to distinguish between the cases $p_S=q_S$
and $\|p_S-q_S\|_2^2 \geq k^{-1}\epsilon^2/(1+O(\alpha m /\sqrt{k}))$ and return the result.

\end{enumerate}
}}

\vspace{0.3cm}

The analysis is quite simple. Firstly, we can assume that $\|p_S\|_2^2 = O(1/m)$ as this happens with $90\%$ probability over the choice of $S$. 
Next, let $T$ be the set of size at most $k$ such that $d_{k,\alpha}(p,q) = \sum_{i\in T} |p(i)-q(i)|$. 
With $90\%$ probability over the choice of $S$, we have that only $O(mk\alpha)$ elements from $S$ land in $T$. 
Assuming this is the case, it is sufficient to distinguish between $p_S=q_S$ and 
$\|p_S-q_S\|_2^2 \geq k^{-1}\epsilon^2/(1+O(\alpha m ))$, 
which can be done in $O(k\epsilon^{-2}(1+O(\alpha m /\sqrt{k}))/\sqrt{m}) = O(k^{2/3}\epsilon^{-4/3}(1+O(\alpha m )))$ samples. 
This completes the proof.
\end{proof}

We are now prepared to prove Theorem \ref{thm:main-alg}. The basic idea 
behind the improvement is that we want to avoid merging heavy bins. 
We do this by first taking a large set of elements and 
defining the $p^{(i)}$ in a way that doesn't involve merging elements of these sets.

\begin{proof}
We first note that given the algorithm from \cite{DKN:15:FOCS}, 
it suffices to provide an algorithm when $\eps > k^{-3/8}$ and $n \leq 2^k$.

Our algorithm is the following:

\vspace{0.3cm}

\fbox{\parbox{6in}{
{\bf Algorithm} \texttt{Small-Domain-$\mathcal{A}_k$-tester}\\
Input: sample access to pdf's $p, q: [n] \to [0, 1]$, $k \in \Z_+$, and $\eps > 0$.\\
Output: ``YES'' if $q = p$; ``NO'' if $\|q-p\|_{\mathcal{A}_k} \ge \eps.$


\begin{enumerate}

\item Let $m=k^{2/3}\log^{4/3}(3+n/k)/\eps^{4/3}.$ Let $C$ be a sufficiently large constant.

\item Let $S$ be a set of $Cm\log(k)$ independent samples from $p$.

\item Let $p^{(0)}=p$ and $q^{(0)}=q$. For $i:=1$ to $t \eqdef \lceil \log_2(n/k)\rceil$, define distributions $p^{(i)},q^{(i)}$ inductively as follows:
\begin{enumerate}
\item $p^{(i)}$ will be a flattening of $p$ by merging all bins in certain dyadic intervals (i.e., intervals of the form $[a\cdot 2^b+1,(a+1)2^b]$).
\item $p^{(i+1)}$ is obtained from $p^{(i)}$ by merging any pair of adjacent bins of $p^{(i)}$ that correspond to intervals 
$[a 2^{i+1}+1,a 2^{i+1}+2^i]$ and $[a2^{i+1}+2^i+1,(a+1)2^i]$ where neither of these subintervals contains a point of $S$.
\item $q^{(i+1)}$ is obtained by merging bins in a similar way.
\end{enumerate}

\item Take $Cm\log\log(3+n/k)$ samples,
and use these samples to distinguish between the cases
$p^{(i)}=q^{(i)}$ and $d_{k,1/m}(p^{(i)}-q^{(i)}) > \epsilon/(8 \log_2(3+n/k))$
with probability of error at most $1/(10 \log_2(3+n/k))$ for each $i$ from $0$ to $t$, 
using the same samples for each test.

\item If any test yields that $p^{(i)} \neq q^{(i)}$, return ``NO''.

\item Otherwise, test if $p^{(t)}=q^{(t)}$ of $\|p^{(t)}-q^{(t)}\|_{\mathcal{A}_k} > \eps/2$ 
using the algorithm from Proposition \ref{algorithmProp} and return the answer.

\end{enumerate}
}}

\vspace{0.3cm}

We now proceed with the analysis. Firstly, we note that the bins of $p^{(t)}$ corresponds to a dyadic interval 
either containing an element of $S$ or adjacent to such an element. 
Therefore, the domain of $p^{(t)}$ is at most $O(t|S|) = \mathrm{poly}(k)$.

We also note that the sample complexity of
\begin{align*}
O(m\log(k))&+O(k^{2/3}\log^{4/3}(3+n/k)\log\log(3+n/k)/\eps^{4/3})\\ 
& + O((k^{2/3}\log^{4/3}(k)\log\log(k)/\eps^{4/3}+\sqrt{k}\log^2(k)\log\log(k)/\eps^2)),
\end{align*}
which is sufficient.

We now proceed to prove correctness. 
For completeness, if $p=q$, it is easy to see that $p^{(i)}=q^{(i)}$ for all $i$, 
and thus, by a union bound, we pass every test and our algorithm returns ``YES'' with $2/3$ probability. 

It remains to consider the soundness case, i.e., the case where $\|p-q\|_{\mathcal{A}_k} > \eps$.
In this case, let $\mathcal{I}=\{I_i\}_{1\leq i \leq k}$ be a partition of $[n]$ into intervals 
so that $\sum_{i=1}^k |p(I_i)-q(I_i)| > \eps$. We claim that with high probability over the choice of $S$ 
every dyadic interval that has mass (under $p$) at least $1/m$ and contains an endpoint of some $I_i$ 
also contains an element of $S$. To prove this, we note that the $I_i$ contain only $O(k)$ endpoints, 
and each endpoint is contained in a unique minimal dyadic interval of mass at least $1/m$. 
It suffices to show that each of these $O(k)$ intervals of mass at least $1/m$ contains a point in $S$, 
but this follows easily by a union bound. Henceforth, we will assume that the $S$ we chose has this property.

Let $\mathcal{I}^{(i)}$ be a partition of the bins for $p^{(i)}$ and $q^{(i)}$ defined inductively 
by $\mathcal{I}^{(0)}=\mathcal{I}$ and $\mathcal{I}^{(i+1)}$ is obtained from $\mathcal{I}^{(i)}$ 
by flattening it and assigning new bins that partially overlap two of the intervals in $\mathcal{I}^{(i)}$ arbitrarily 
to one of the two corresponding intervals in $\mathcal{I}^{(i+1)}$.

We note that
$$
\left|\sum_{I\in \mathcal{I}^{(i)}}|p^{(i)}(I)-q^{(i)}(I)|- \sum_{I\in \mathcal{I}^{(i+1)}}|p^{(i+1)}(I)-q^{(i+1)}(I)|\right|
$$
is at most twice a sum over $k$ bins $b$, 
not containing an element of $S$ of $|p^{(i)}(b)-q^{(i)}(b)|$. This in turn is at most $2d_{k,1/m}(p^{(i)},q^{(i)})$. 
Inducting, we have that
$$
\|p-q\|_{\mathcal{A}_k} \leq 2\sum_{i=0}^{t-1}d_{k,1/m}(p^{(i)},q^{(i)}) + \|p^{(t)}-q^{(t)}\|_{\mathcal{A}_k}.
$$
Therefore, if $\|p-q\|_{\mathcal{A}_k}>\eps$, we have that either $d_{k,1/m}(p^{(i)}-q^{(i)}) > \epsilon/(8 \log_2(3+n/k))$ for some $i$, 
or $\|p^{(t)}-q^{(t)}\|_{\mathcal{A}_k}>\eps/2$. In either case, with probability at least $2/3$, our algorithm will detect this and reject.
This completes the proof.
\end{proof}

\section{Nearly Matching Information-Theoretic Lower Bound}

In this section, we prove a nearly matching sample lower bound. 
We first show a slightly easier lower bound that holds even for distributions that are piecewise constant on a few pieces, 
and then modify it to obtain the stronger general bound for testing closeness in $\mathcal{A}_k$ distance.

\subsection{Lower Bound for $k$-Histograms}

We begin with a lower bound for $k$-histograms ($k$-flat distributions). 
Before moving to the discrete setting, we first establish a lower bound 
for continuous histogram distributions.
Our bound on discrete distributions will follow
from taking the adversarial distribution from this example
and rounding its values to the nearest integer.
In order for this to work, we will need ensure to that our adversarial
distribution does not have its $\mathcal{A}_k$-distance decrease
by too much when we apply this operation.
To satisfy this requirement, we will guarantee
that our distributions will be piecewise constant with all the pieces of length at least $1$.

\begin{proposition} \label{prop:lb}
Let $k \in \Z_+$, $\epsilon>0$ sufficiently small, and $W>2$ .
Fix $$m=\min(k^{2/3}\log^{1/3}(W)/\epsilon^{4/3}, k^{4/5}/\epsilon^{6/5}) \;.$$
There exist distributions $\mathcal{D},\mathcal{D'}$ over pairs of distributions
$p$ and $q$ on $[0, 2(m+k)W]$, where $p$ and $q$ are $O(m+k)$-flat
with pieces of length at least $1$, so that: (a) when drawn from $\mathcal{D}$,
we have $p=q$ deterministically, (b) when drawn from $\mathcal{D'}$, we have 
$\|p-q\|_{\mathcal{A}_k}>\epsilon$ with $90\%$ probability,
and so that $o(m)$ samples are insufficient to distinguish whether
or not the pair is drawn from $\mathcal{D}$ or $\mathcal{D'}$
with better than $2/3$ probability.
\end{proposition}

At a high-level, our lower bound construction proceeds as follows: 
We will divide our domain into $m+k$ bins
so that no information about which distributions had samples drawn
from a given bin or the ordering of these samples will help to distinguish
between the cases of $p=q$ and otherwise,
unless at least three samples are taken from the bin in question.
Approximately $k$ of these bins will each have mass $\epsilon/k$
and might convey this information if at least three samples are taken from the bin.
However, the other $m$ bins will each have mass approximately $1/m$
and will be used to add noise. In all, if we take $s$ samples,
we expect to see approximately $s^3\epsilon^3/k^2$ of the lighter
bins with at least three samples.
However, we will see approximately $s^3/m^2$ of our heavy bins with three samples.
In order for the signal to overwhelm the noise, 
we will need to ensure that we have $(s^3\epsilon^3/k^2)^2 > s^3/m^2$.

The above intuitive sketch assumes that we cannot obtain information
from the bins in which only two samples are drawn.
This naively should not be the case.
If $p=q$, the distance between two samples drawn from that bin will be independent
of whether or not they are drawn from the same distribution.
However, if $p$ and $q$ are supported on disjoint intervals,
one would expect that points that are close to each other should be
far more likely to be drawn from the same distribution than from different distributions.
In order to disguise this, we will scale the length of the intervals by a random, exponential amount,
essentially making it impossible to determine what is meant by two points being close to each other.
In effect, this will imply that two points drawn from the same bin will only reveal $O(1/\log(W))$
bits of information about whether $p=q$ or not.
Thus, in order for this information to be sufficient, 
we will need that $(s^2\epsilon^2/k)^2 /\log(W) > (s^2/m)$.
We proceed with the formal proof below.

\begin{proof}[Proof of Proposition~\ref{prop:lb}:]
We use ideas from~\cite{DK16} to obtain this lower bound 
using an information theoretic argument.

We may assume that $\epsilon >k^{1/2}$,
because otherwise we may employ the standard lower bound that $\Omega(\sqrt{k}/\epsilon^2)$
samples are required to distinguish two distributions on a support of size $k$.

First, we note that it is sufficient to take $\mathcal{D}$ and $\mathcal{D'}$ be 
distributions over pairs of non-negative, piecewise constant distributions
with total mass $\Theta(1)$ with $90\%$ probability
so that running a Poisson process with parameter $o(m)$
is insufficient to distinguish a pair from $\mathcal{D}$ from a pair from $\mathcal{D'}$~\cite{DK16}.

We construct these distributions as follows:
We divide the domain into $m+k$ bins of length $2W$.
For each bin $i$, we independently generate a random $\ell_i$,
so that $\log(\ell_i/2)$ is uniformly distributed over $[0,2\log(W)/3]$.
We then produce an interval $I_i$ within bin $i$ of total length $\ell_i$
and with random offset. In all cases, we will have $p$ and $q$ supported
on the union of the $I_i$'s.

For each $i$ with probability $m/(m+k)$, we have the restrictions of $p$ and $q$ to $I_i$
both uniform with $p(I_i)=q(I_i)=1/m$.
The other $k/(m+k)$ of the time we have $p(I_i)=q(I_i)=\epsilon/k$.
In this latter case, if $p$ and $q$ are being drawn from $\mathcal{D}$,
$p$ and $q$ are each constant on this interval.
If they are being drawn from $\mathcal{D'}$,
then $p+q$ will be constant on the interval,
with all of that mass coming from $p$
on a random half and coming from $q$ on the other half.

Note that in all cases $p$ and $q$ are piecewise constant
with $O(m+k)$ pieces of length at least $1$.
It is easy to show that with high probability
the total mass of each of $p$ and $q$ is $\Theta(1)$,
and that if drawn from $\mathcal{D'}$
that $\|p-q\|_{\mathcal{A}_k}\gg \epsilon$ with at least $90\%$ probability.

We will now show that if one is given $m$ samples from each of $p$ and $q$,
taken randomly from either $\mathcal{D}$ or $\mathcal{D'}$,
that the shared information between the samples and the source family will be small.
This implies that one is unable to consistently guess whether our pair was taken from $\mathcal{D}$ or $\mathcal{D'}$.

Let $X$ be a random variable that is uniformly at random
either $0$ or $1$. Let $A$ be obtained by applying
a Poisson process with parameter $s=o(m)$
on the pair of distributions $p,q$ drawn from $\mathcal{D}$
if $X=0$ or from $\mathcal{D'}$ if $X=1$.
We note that it suffices to show that the shared information $I(X:A)=o(1)$.
In particular, by Fano's inequality, we have:
\begin{lemma}\label{informationTheoryLem}
If $X$ is a uniform random bit and $A$ is a correlated random variable,
then if $f$ is any function so that $f(A)=X$ with at least $51\%$ probability,
then $I(X:A)\geq 2\cdot 10^{-4}$.
\end{lemma}

Let $A_i$ be the samples of $A$ taken from the $i^{th}$ bin.
Note that the $A_i$ are conditionally independent on $X$.
Therefore, we have that
$
I(X:A) \leq \sum_i I(X:A_i) = (m+k)I(X:A_1) \;.
$
We will proceed to bound $I(X:A_1)$.

We note that $I(X:A_1)$ is at most the integral
over pairs of multisets $a$ (representing a set of samples
from $q$ and a set of samples from $p$),  of
$$
O\left(\frac{(\Pr(A_1=a|X=0)-\Pr(A_1=a|X=1))^2}{\Pr(A_1=a)} \right).
$$
Thus,
$$
I(X:A_1)=\sum_{h=0}^\infty \int_{|a|=h}O\left(\frac{(\Pr(A_1=a|X=0)-\Pr(A_1=a|X=1))^2}{\Pr(A_1=a)} \right).
$$
We will split this sum up based on the value of $h$.

For $h=0$, we note that the distributions for $p+q$
are the same for $X=0$ and $X=1$. Therefore,
the probability of selecting no samples is the same.
Therefore, this contributes $0$ to the sum.

For $h=1$, we note that the distributions for $p+q$
are the same in both cases, and conditioning on
$I_1$ and $(p+q)(I_1)$ that $\E[p]$ and $\E[q]$
are the same in each of the cases $X=0$ and $X=1$.
Therefore, again in this case, we have no contribution.

For $h\geq 3$, we note that
$
I(X:A_1) \leq I(X:A_1,I_1) \leq I(X:A_1|I_1) \;,
$
since $I_1$ is independent of $X$.
We note that $\Pr(A_1=a|X=0,p(I_1)=1/m)=\Pr(A_1=a|X=1,p(I_1)=1/m)$.
Therefore, we have that
$$
\Pr(A_1=a|X=0)-\Pr(A_1=a|X=1) = \Pr(A_1=a|X=0,p(I_1)=\epsilon/k)-\Pr(A_1=a|X=1,p(I_1)=\epsilon/k).
$$
If $p(I_1)=\epsilon/k$, the probability that exactly $h$ elements
are selected in this bin is at most $k/(m+k)(2s\epsilon/k)^h/h!$,
and if they are selected, they are uniformly distributed in $I_1$
(although which of the sets $p$ and $q$ they are taken from is non-uniform).
However, the probability that $h$ elements are taken from $I_1$ is at least
$\Omega(m/(m+k) (sm)^{-h}/h!)$ from the case where $p(I_1)=1/m$,
and in this case the elements are uniformly distributed
in $I_1$ and uniformly from each of $p$ and $q$.
Therefore, we have that this contribution to our shared information is at most
$
k^2/(m(m+k)) O(s \epsilon^2 m/k^2)^h/h! \;.
$
We note that $\epsilon^2 m / k^2 < 1$.
Therefore, the sum of this over all $h\geq 3$ is
$
k^2/(m(m+k)) O(s\epsilon^2 m/k^2)^3.
$
Summing over all $m+k$ bins, this is
$
k^{-4}\epsilon^6 s^3m^2 = o(1).
$

It remains to analyze the case where $h=2$.
Once again, we have that ignoring which of $p$ and $q$
elements of $A_1$ came from, $A_1$
is identically distributed conditioned on $p(I_1)=1/m$ and $|A_1|=2$
as it is conditioned on $p(I_1)=\epsilon/k$ and $|A_1|=2$.
Since once again, the distributions $\mathcal{D}$ and $\mathcal{D'}$
are indistinguishable in the former case, we have that the contribution
of the $h=2$ terms to the shared information is at most
$$
O\left(\frac{(k/(k+m) (\epsilon s/k)^2)^2}{m/(k+m)(s/m)^2} \right)\dtv((A_1|X=0,p(I_1)\epsilon/k,|A_1|=2),(A_1|X=1,p(I_1)=\epsilon/k,|A_1|=2))
$$
or
$$
O\left(s^2mk^{-2}\epsilon^4/(k+m) \right)\dtv((A_1|X=0,p(I_1)=\epsilon/k,|A_1|=2),(A_1|X=1,p(I_1)=\epsilon/k,|A_1|=2)) \;.
$$
It will suffice to show that conditioned
upon $p(I_1)=\epsilon/k$ and $|A_1|=2$ that $$\dtv((A_1|X=0),(A_1|X=1))=O(1/\log(W)).$$
Let $f$ be the order preserving linear function from $[0,2]$ to $I_1$.
Notice that conditional on $|A_1|=2$ and $p(I_1)=\epsilon/k$
that we may sample from $A_1$ as follows:
\begin{itemize}
\item Pick two points $x>y$ uniformly at random from $[0,2]$.
\item Assign the points to $p$ and $q$ as follows:
\begin{itemize}
\item If $X=0$ uniformly randomly assign these points to either distribution $p$ or $q$.
\item If $X=1$ randomly do either:
\begin{itemize}
\item Assign points in $[0,1]$ to $q$ and other points to $p$.
\item Assign points in $[0,1]$ to $p$ and other points to $q$.
\end{itemize}
\end{itemize}
\item Randomly pick $I_1$ and apply $f$ to $x$ and $y$ to get outputs $z=f(x),w=f(y)$.
\end{itemize}

Notice that the four cases: (i) both points coming from $p$,
(ii) both points coming from $q$, (iii) a point from $p$ preceding a point from $q$,
(iv) a point from $q$ preceding a point from $p$,
are all equally likely conditioned on either $X=0$ or $X=1$.
However, we will note that this ordering is no longer independent of the choice of $x$ and $y$.

Therefore, we can sample from $A_1$ subject to $X=0$ and
from $A_1$ subject to $X=1$ in such a way that this ordering
is the same deterministically. We consider running the above
sampling algorithm to select $(x,y)$ while sampling
from $X=0$ and $(x',y')$ when sampling from $X=1$
so that we are in the same one of the above four cases. We note that
$$
\dtv((A_1|X=0),(A_1|X=1)) \leq \E_{x,y,x',y'}[\dtv((f(x),f(y)),(f(x'),f(y')))] \;,
$$
where the variation distance is over the random choices of $f$.

To show that this is small, we note that $|f(x)-f(y)|$ is distributed like $\ell_1(x-y)$.
This means that $\log(|f(x)-f(y)|)$ is uniform over
$[\log(f(x)-f(y)),\log(f(x)-f(y))+2\log(W)/3]$.
Similarly, $\log(|f'(x')-f'(y')|)$ is uniform over $[\log(f(x')-f(y')),\log(f(x')-f(y'))+2\log(W)/3]$.
These differ in total variation distance by
$$
O\left(\frac{|\log(f(x)-f(y))|+|\log(f(x')-f(y'))|}{\log(W)}\right) \;.
$$
Taking the expectation over $x,y,x',y'$ we get $O(1/\log(W))$.
Therefore, we may further correlate the choices made
in selecting our two samples,
so that $z-w=z'-w'$ except with probability $O(1/\log(W))$.
We note that after conditioning on this, $z$ and $z'$
are both uniformly distributed over subintervals of $[0,2W]$
of length at least $2(W-W^{2/3})$. Therefore, the distributions
on $z$ and $z'$ differ by at most $O(W^{-1/3}).$
Hence, the total variation distance between $A_1$ conditioned
on $|A_1|=2,p(I_1)=\epsilon/k,X=0$
and conditioned on $|A_1|=2,p(I_1)=\epsilon/k,X=1$
is at most $O(1/\log(W))+O(W^{-1/3})=O(1/\log(W))$.
This completes our proof.
\end{proof}

We can now turn this into a lower bound
for testing $\mathcal{A}_k$ distance on discrete domains.


\begin{proof}[Proof of second half of Theorem~\ref{thm:main-lb}:]
Assume for sake of contradiction that this is not the case,
and that there exists a tester taking $o(m)$ samples.
We use this tester to come up with a continuous tester that violates Proposition~\ref{prop:lb}.

We begin by proving a few technical bounds on the parameters involved.
Firstly, note that we already have a lower bound of $\Omega(k^{1/2}/\epsilon^2)$,
so we may assume that this is much less than $m$.
We now claim that $m = O(\min(k^{2/3}\log^{1/3}(3+n/(m+k))/\epsilon^{4/3},k^{4/5}/\epsilon^{6/5}).$
If $m\leq k$, there is nothing to prove. Otherwise,
$$
k^{2/3}\log^{1/3}(3+n/(m+k))/\epsilon^{4/3} \geq m(m/k)^{-1/3} \log(3+n/(m+k))^{1/3}.
$$
Thus, there is nothing more to prove unless
$\log(3+n/(m+k))\gg m/k$.
But, in this case, $\log(3+n/(m+k)) \gg \log(m/k)$
and thus $\log(3+n/(m+k))=\Theta(\log(3+n/k))$, and we are done.

We now let $W=n/(6(m+k))$, and let $\mathcal{D}$ and $\mathcal{D'}$
be as specified in Proposition~\ref{prop:lb}.
We claim that we have a tester to distinguish a $p,q$ from $\mathcal{D}$
from ones taken from $\mathcal{D'}$ in $o(m)$ samples.
We do this as follows: By rounding $p$ and $q$ down to the nearest
third of an integer, we obtain $p'$,$q'$ supported on set of size $n$.
Since $p$ and $q$ were piecewise constant on pieces of size at least $1$,
it is not hard to see that $\|p'-q'\|_{\mathcal{A}_k} \geq \|p-q\|_{\mathcal{A}_k}/3.$
Therefore, a tester to distinguish $p'=q'$ from $\|p'-q'\|_{\mathcal{A}_k}\geq \epsilon$
can be used to distinguish $p=q$ from $\|p-q\|_{\mathcal{A}_k} \geq 3\epsilon.$
This is a contradiction and proves our lower bound.
\end{proof}

\subsection{The Stronger Lower Bound}

In order to improve on the bound from the last section, we will need to modify our previous construction in two ways both having to do with the contribution to the shared information coming from the case where two samples are taken from the same bin. The first is that we will need a different way of distinguishing between $\mathcal{D}$ and $\mathcal{D'}$ so that the variation distance between the distributions obtained from taking a pair of samples from the same bin is $O(1/\log^2(W))$ rather than $O(1/\log(W))$. After that, we will also need a better method of disguising these errors. In particular, in the current construction, most of the information coming from pairs of samples from the same bin occurs when the two samples are very close to each other (as when this happens in $\mathcal{D'}$, the samples usually don't come one from $p$ and the other from $q$). This is poorly disguised by noise coming from the heavier bins since these are not particularly likely to produce samples that are close. We can improve our way of disguising this by having different heavy bins to better mask this signal.

In order to solve the first of these problems, we will need the following construction:
\begin{lemma}\label{distanceHideLem}
Let $W$ be a sufficiently large integer. 
There exists a family $\mathcal{E}$ of pairs of distributions $p$ and $q$ on $[0,W]$ so that the following holds:

Firstly, $p$ and $q$ are deterministically supported on disjoint intervals, and thus have $\mathcal{A}_1$ distance $2$. 
Furthermore, let $\mathcal{E}_0$ be the family of pairs of distributions $p$ and $q$ on $[W]$ obtained by taking 
$(p',q')$ from $\mathcal{E}$ and letting $p=q=(p'+q')/2$. In other words, a sample from $\mathcal{E}_0$ 
can be thought of as taking a sample from $\mathcal{E}$ and then re-randomizing the label. 
Consider the distribution obtained by sampling $(p,q)$ from $\mathcal{E}$, and then taking 
two independent samples $x$ and $y$ from $(p+q)/2$. 
We let $\mathcal{E}^2$ be the induced distribution on $x$ and $y$ along with the labels 
of which of $p$ and $q$ each were taken from. Define $\mathcal{E}_0^2$ similarly, and note 
that it is equivalent to taking a sample from $\mathcal{E}^2$ and re-randomizing the labels. 
Then $\dtv(\mathcal{E}^2,\mathcal{E}_0^2)=O(1/\log^2(W)).$
\end{lemma}
\begin{proof}
We note that it is enough to construct a family of continuous distributions $p$ and $q$ on $[0,W]$ 
so that deterministically $p$ and $q$ are supported on intervals separated by distance $2$, 
and so that the second condition above holds. By then rounding the values of $p$ and $q$ to the nearest integer, 
we obtain an appropriate discrete distribution.

The construction of $\mathcal{E}$ is straightforward. First, choose $a$ uniformly from $[W^{2/3},W-W^{2/3}]$, 
$\ell$ uniformly from $[0,\log(W)/3]$, and $b$ uniformly from $\{\pm 1\}$. To sample from $p$, 
take an $\alpha$ uniformly from $[0,\log(W)/3]$ and return $a+b e^\ell e^\alpha$. 
To sample from $q$, take an $\alpha$ uniformly from $[0,\log(W)/3]$ and return $a-b e^\ell e^\alpha$.

It is clear that $p$ and $q$ are supported on disjoint intervals of distance at least $2$. 
It remains to prove the more complicated claim.

Let $\mathcal{E}_s^2$ be the distribution obtained by picking a pair of distributions from $\mathcal{E}$ 
and then returning two independent samples from $p$. Let $\mathcal{E}_d^2$ be the distribution 
obtained by picking a pair of distributions from $\mathcal{E}$ and then returning independent 
samples from $p$ and $q$. We claim that $\dtv(\mathcal{E}^2,\mathcal{E}_0^2)=O(\dtv(\mathcal{E}_s^2,\mathcal{E}_d^2)).$ 
This is because if a sample from $\mathcal{E}^2$ has both points coming from $p$ or both from $q$, 
the points come from $\mathcal{E}_s^2$, whereas if one point comes from each, 
the points come from $\mathcal{E}_d^2$. On the other hand, in any of these cases, 
a pair of samples from $\mathcal{E}_0^2$ comes from $(\mathcal{E}_s^2+\mathcal{E}_d^2)/2$.

Let $(x,y)$ be a sample from $\mathcal{E}_s^2$ and $(w,z)$ a sample from $\mathcal{E}_d^2$. 
We claim that $\dtv((x,y),(w,z))\leq \dtv(x-y,w-z)+O(W^{-1/3})$. 
This is because of the averaging over $a$ in the definition of $\mathcal{E}$. 
In particular, consider the following mechanism for taking a sample from $\mathcal{E}_s^2$ or $\mathcal{E}_d^2$. 
First, randomly select values of $s$ and $\ell$. Then select the $\alpha$ and $\alpha'$ for the two sample points. 
Finally, sample the defining value of $a$. Notice that the difference between the two final points 
does not depend on the choice of $a$. In fact, after making all other choices, the final distribution is within $O(W^{-1/3})$ 
of the uniform distribution over pairs of points in $[0,W]$ with this distance. 
Thus, $(x,y)$ is close distributionally to the distribution on pairs in $[0,W]$ with separation $x-y$. 
A similar statement holds for $(z,w)$ and points with separation $z-w$. Thus, $\dtv((x,y),(w,z))=\dtv(x-y,w-z)+O(W^{-1/3})$, as desired.

Next, we claim that $\dtv(x-y,z-w)=\dtv(|x-y|,|z-w|)$. 
This is easily seen to be the case by averaging over $b$. 
We have left to bound the latter distance. If $x$ and $y$ are chosen using $\alpha_x$ and $\alpha_y$, 
we have that $|x-y| = e^\ell |e^{\alpha_x}-e^{\alpha_y}|$. 
Similarly, if $z$ and $w$ are chosen using $\alpha_z$ and $\alpha_w$, 
we have that $|z-w| = e^\ell |e^{\alpha_z} + e^{\alpha_w}|$. 
Notice that if we fix $\alpha_x,\alpha_y,\alpha_z$ and $\alpha_w$, 
the variation distance between these two distributions (given the distributions over the values of $\ell$) is
$$
O\left(\frac{\left|\log\left(\frac{|e^{\alpha_x}-e^{\alpha_y}|}{|e^{\alpha_z}+e^{\alpha_w}|} \right) \right|}{\log(W)} \right).
$$
Therefore, the variation distance between $|x-y|$ and $|z-w|$ is $O(1/\log(W))$ 
times the earth mover distance between $\log(|e^{\alpha_x}-e^{\alpha_y}|)$ and $\log(|e^{\alpha_z}+e^{\alpha_w}|)$. 
Correlating these variables so that $\alpha_x=\alpha_z=\alpha$ and $\alpha_y=\alpha_w=\beta$, 
this is at most the expectation of $|\log(\tanh((\alpha-\beta)/2))|$, which can easily be seen to be $O(1/\log(W))$. 
This shows that $\dtv(\mathcal{E}^2,\mathcal{E}_0^2)=O(1/\log^2(W))$, completing our proof.
\end{proof}

We are now ready to prove the first part of Theorem \ref{thm:main-lb}.
\begin{proof}
The overall outline is very similar to the methods used in the last section. 
For sufficiently large integers $m,k,W$ and $\epsilon>0$ we are going to define families of pairs of pseudo-distributions 
$\mathcal{D}$ and $\mathcal{D'}$ on $[(k+2m)W]$ so that:
\begin{itemize}
\item With $90\%$ probability a random sample from either $\mathcal{D}$ or $\mathcal{D'}$ consists 
of two pseudo-distributions with total mass $\Theta(1)$.
\item The distributions picked by a sample from $\mathcal{D}$ are always the same.
\item The two distributions picked by a sample from $\mathcal{D'}$ have $\mathcal{A}_k$ distance $\Omega(\eps)$ with $90\%$ probability.
\item Letting $A$ be the outcome of a Poisson process with parameter $m$ run on a random sample from either $\mathcal{D}$ or $\mathcal{D'}$, the family used cannot be reliably determined from $A$ unless $m \gg k^{4/5}/\eps^{6/5}$ 
or $m\gg k^{2/3}\log^{4/3}(W)/\eps^{4/3}$.
\end{itemize}

Before we define $\mathcal{D}$ and $\mathcal{D'}$, we will need to define one more family. 
Firstly, let $\mathcal{E}$ and $\mathcal{E}_0$ be the families of distributions on $[W]$ 
from Lemma \ref{distanceHideLem}. Let $\mathcal{E}^2$ and $\mathcal{E}_0^2$ be as described in that lemma. 
We define another family, $\mathcal{F}$ of pairs of distributions on $[W]$ as follows. 
First select a point $(x,y)$ from the renormalized version of $|\mathcal{E}^2-\mathcal{E}_0^2|$. 
Then return the pair of distributions $p=q$ equals the uniform distribution over $\{x,y\}$.

To define $\mathcal{D}$ and $\mathcal{D'}$, we split $[(k+2m)W]$ into $k+2m$ blocks of size $W$. 
A sample from $\mathcal{D}$ assigns to each block independently the pseudo-distribution:
\begin{itemize}
\item $\mathcal{E}_0/m$ (i.e., a random sample from $\mathcal{E}_0$ scaled by a factor of $1/m$) with probability $m/(k+2m)$
\item $\mathcal{E}_0\epsilon/k$ with probability $k/(k+2m)$
\item $\mathcal{F}/m$ with probability $m/(k+2m)$.
\end{itemize}
A sample from $\mathcal{D'}$ assigns to each block independently the pseudo-distribution:
\begin{itemize}
\item $\mathcal{E}_0/m$ with probability $m/(k+2m)$
\item $\mathcal{E}\epsilon/k$ with probability $k/(k+2m)$
\item $\mathcal{F}/m$ with probability $m/(k+2m)$.
\end{itemize}

It is easy to see that $\mathcal{D}$ and $\mathcal{D'}$ satisfy the first three of the properties listed above. 
To demonstrate the fourth, let $X$ be a uniform Bernoulli random variable. 
Let $A$ be obtained by applying a Poisson process of parameter $m$ to a  sample from $\mathcal{D}$ if $X=0$, 
and to a sample from $\mathcal{D'}$ if $X=1$. We will show that $I(X:A)=o(1).$ 
Once again, letting $A=(A_1,A_2,\ldots,A_{k+2m})$, where $A_i$ are the samples taken from the $i^{th}$ block, 
we note that the $A_i$ are conditionally independent on $X$ and therefore, $I(X:A)\leq (k+2m) I(X:A_1)$.

As before, no information is gained when $|A_1|<2$, and the contribution when $|A_1|\geq 3$ is $O((k/(k+m))^2 (m \eps/k)^6/ (m/(k+m)))$, 
which leads to a total contribution of $o(1)$ when $m=o(k^{4/5}/\eps^{6/5})$. 
It remains to consider the contribution from events where $|A_1|=2$.

This is
$$
\sum_{x\in([W]\times \{p,q\})^2} O\left(\frac{(\Pr(A_1=x|X=0)-\Pr(A_1=x|X=1))^2}{\Pr(A_1=x)} \right).
$$
Note that the contribution to
$
\Pr(A_1=x|X=0)-\Pr(A_1=x|X=1)
$
from cases where $\mathcal{D}$ and $\mathcal{D'}$ 
on block $1$ are $\mathcal{E}_0/m$ or $\mathcal{F}/m$ cancel out. 
Therefore, we have that
\begin{align*}
|\Pr(A_1=x|X=0)-\Pr(A_1=x|X=1)| & = O((k/(k+m)) (m\eps/k)^2 |\mathcal{E}^2(x)-\mathcal{E}_0^2(x)|)\\ & = O((k/(k+m))(m\eps/k)^2 \mathcal{F}(x) / \log^2(W)).
\end{align*}
On the other hand, the $\Pr(A_1=x)$ is at least the probability that $A_1=x$ when the restriction to block $1$ is $\mathcal{F}/m$, 
which is $\Omega(m/(k+m)\mathcal{F}(x))$. Therefore, the contribution to $I(X:A_1)$ coming from events where $|A_1|=2$ is
\begin{align*}
&  \sum_{x\in([W]\times \{p,q\})^2} O\left(\frac{(\Pr(A_1=x|X=0)-\Pr(A_1=x|X=1))^2}{\Pr(A_1=x)} \right)\\
& = \sum_{x\in([W]\times \{p,q\})^2}O\left(\frac{((k/(k+m))(m\eps/k)^2 \mathcal{F}^2(x) / \log^2(W))^2}{m/(k+m)\mathcal{F}^2(x)} \right)\\
& = \sum_{x\in([W]\times \{p,q\})^2}\mathcal{F}^2(x) O(k^2 (m\eps/k)^4 \log^{-4}(W)/ (m(k+m)))\\
& = O(m^3 \eps^4 k^{-2} \log^{-4}(W) / (m+k)).
\end{align*}
Hence, the total contribution to $I(X:A)$ from such terms is $O(m^3 \eps^4 k^{-2} \log^{-4}(W) / (m+k)).$ This is $o(1)$ if $m=o(k^{2/3}\log^{4/3}(W)/\eps^{4/3}).$ This completes our proof.
\end{proof}

\bibliographystyle{alpha}
\bibliography{allrefs}

\end{document}